\documentclass[%
 reprint,
 amsmath,amssymb,
 aps,
]{revtex4-2}
\usepackage{amsthm}

\usepackage{import}
\usepackage{physics}
\newtheorem{theorem}{Theorem}
\usepackage[utf8]{inputenc}
\usepackage{lineno,hyperref}
\usepackage{graphicx}

\newtheorem{lemma}[theorem]{Lemma}
\newtheorem{definition}{Definition}[section]
\usepackage{dcolumn}
\usepackage{bm}
\usepackage{bbm} 
\usepackage{xcolor}

\begin{document}


\title{Quantum capacity of a deformed bosonic dephasing channel}

\author{Shahram Dehdashti$^{1}$}
\email{shahram.dehdashti@tum.de}
\author{Janis N\"{o}tzel$^{1}$}
\email{janis.noetzel@tum.de}
\author{Peter van Loock$^{2}$}
\email{loock@uni-mainz.de}
\affiliation{$^{1}$Emmy-Noether Gruppe Theoretisches Quantensystemdesign Lehrstuhl F\"{u}r Theoretische Informationstechnik
Technische Universit\"{a}t M\"{u}nchen.}
\affiliation{$^{2}$Johannes-Gutenberg University of Mainz, Institute of Physics, Staudingerweg 7, 55128 Mainz, Germany}





\begin{abstract}
In this paper, using the notion of nonlinear coherent states, we define a deformed bosonic dephasing channel modelling the impact of a Kerr medium on a quantum state, as it occurs, for instance, in quantum communication based on optical fibers. We show that, in certain regimes, the Kerr nonlinearity is able to compensate the dephasing.
In addition, our studies reveal that the quantum capacity of the deformed bosonic dephasing channel can be greater than that of the undeformed, standard bosonic dephasing channel for certain nonlinearity parameters.
\end{abstract}

\maketitle

\section{Introduction }
One of the great aims of quantum information science is to encode and process information coherently \cite{1} between several subsystems, a capability which enables quantum algorithms to, for instance, factor large integers \cite{2,3}, simulate complex physical dynamics \cite{4,4.5}, or solve unstructured search problems with proven speedups \cite{5,6,7}. In addition, information can be transmitted at high speed using the concept of joint detection receivers \cite{8,9,10,11}, with entanglement assistance \cite{12}, or securely over quantum channels \cite{13,14}. However, in practice, huge practical challenges arise. Among these are the control of non-linearities in the quantum devices \cite{15,16,18,19} and the omnipresent decoherence effects \cite{20}. \\
\indent 
The classical theory of communications was developed mostly in the context of linear channels with additive noise, which was adequate for electromagnetic propagation through wires,  cables and air within certain boundary conditions, for example, on the power of the transmitted signal. 
However, since the advent of optical fibers as the backbone of the internet, we are also faced with a non-linear propagation channel. These channels are normally described by a non-linear Schr\"{o}dinger equation, posing major challenges to our understanding. The difficulty is that the input–output relationship of an optical fiber channel is obtained by integrating a non-linear partial differential equation and may not be represented by an instantaneous non-linearity \cite{21}. For optical fiber communication systems, nonlinear interactions have a huge impact on the capacity \cite{22}. While in theory a large enough transmission power would enable the transmission of an arbitrarily large number of bits per second (depending mainly on the input power), the fiber nonlinearities put very practical limits on the transmission power. The question to what degree quantum methods could be used to overcome design limitations in such systems is not only open, but can be answered only based on corresponding system models. In this domain, the work of Ref.~\cite{19} is the first systematic analysis of the fiber nonlinearity in this context. In that work \cite{19}, a system model restricted to coherent-state input is derived. This model clearly shows the decoherence of the coherent-state input. Our model improves upon this earlier work by modelling the impact of the Kerr medium for arbitrary quantum states. 

Decoherence by definition is a process in which a coherent superposition state is reduced to an incoherent probabilistic mixture of the states, i.e., $\sum_{n,m}c_{n}c_{m}^{\ast}\ketbra{\psi_{n}}{\psi_{m}}
\rightarrow \sum_{n}|c_{n}|^{2}\ketbra{\psi_{n}}{\psi_{n}}$. Preventing decoherence is one of the biggest challenges in the quantum domain
\cite{24}. However, to properly protect against decoherence, the ways in which it takes place need to be understood first. As a simple model for decoherence  we can consider a dynamical process by which the inputs are mapped to an output as follows,
\begin{eqnarray}
\sum_{n,m}c_{n}c_{m}^*\ketbra{n}{m}\mapsto  \sum_{n,m}e^{-\gamma(n-m)^{2}/2}c_{n}c_{m}^*\ketbra{n}{m}.
\end{eqnarray}
In this model the decoherence parameter $\gamma>0$  is related to the  strength of the decoherence. In the limit $\gamma\rightarrow \infty$, all off-diagonal components approach zero, while the magnitude of the diagonal components are retained. Such processes have been studied from several different angles in quantum processing \cite{24,25,26,27,28}. 
The above-mentioned transition  can be described for bosonic systems via the so-called bosonic dephasing channel \cite{28,29,30,31,32}
\begin{eqnarray}\label{eqn:genericDephasingChannel}
\rho \mapsto \mathcal{N}_{\gamma}(\hat{\rho})=\sum_{n,m=0}e^{-\frac{1}{2}\gamma(m-n)^{2}}\rho_{m,n}\ketbra{m}{n},
\end{eqnarray}
in which $\rho=\sum_{m,n}\rho_{m,n}\ketbra{m}{n}$. It is possible \cite{28} to derive the bosonic dephasing channel $\mathcal{N}_{\gamma}$ via an interaction between system $S$ and environment $E$ as 
\begin{eqnarray}
\mathcal{N}_{\gamma}=\Tr_{E}\left[ \hat{U} \left( \rho \otimes \ketbra{0}{0} \right)  \hat{U}^{\dagger}   \right],
\end{eqnarray}
where $\rho \in \mathcal{T}(\mathcal{H}_{S})$ is an initial state of the system and $\ket{0}\in \mathcal{H}_{E}$ is a fixed initial state of the environment. The unitary operator $\hat{U}=\exp\left[-\mathbbm{i}\sqrt{\gamma}\hat{a}^{\dagger}\hat{a}\left(\hat{b}+\hat{b}^{\dagger}\right)\right]$ defines the interaction between system and environment and is composed of annihilation and creation operators $\hat{a}$ and $\hat{a}^{\dagger}$  acting on the system Hilbert space and their corresponding counterparts  $\hat{b}$ and $\hat{b}^{\dagger} $ acting on the Hilbert space of the environment. 
Note that while this unitary is generated by a cubic two-mode Hamiltonian, corresponding to a non-linear, non-Gaussian mode transformation and hence (after tracing out the environment $E$) a non-Gaussian channel acting on mode $S$ (the system), the involved phase rotation of $S$ enacted by $\hat U$ depends only linearly on the environmental mode operators and the free evolution of $S$ is, as usual, quadratic in the system's mode operators (or, equivalently, linear in the system's energy).

However, the creation and annihilation operators are derived from the Hamiltonian $\hat H_0=\omega\hat n$ where $\hat n = \hat a^\dag\hat a$ and thus a natural question to ask is how the channel $\mathcal N_\gamma$ behaves when a different Hamiltonian is used in its definition. 

This question is the starting point of our work which improves upon \cite{28} by considering instead of the Hamiltonian $\hat H_0$ the Kerr Hamiltonian, $\hat{H}_{\omega,\lambda}=\omega \hat{n}+\lambda \hat{n}^{2}/2$, where $\hat{n} $ is a number operator, and $\lambda$ defines the power of the non-linearity. For convenience, we make a distinction between the negative and positive values of the non-linear parameter. 
We then decompose $\hat H_\lambda$ as $\hat H_{\omega,\lambda}=\hat A^\dag\hat A$ and use these deformed annihilation- and creation operators $\hat A^\dag$ and $\hat A$ to redefine the unitary interaction between system and environment as 
\begin{align}\label{eqn:kerr-unitary}
    \hat{U}=\exp\left[-\mathbbm{i}\sqrt{\gamma}\hat{A}^{\dagger}\hat{A}\left(\hat{B}+\hat{B}^{\dagger}\right)\right]. 
\end{align}
By using a calculus based on the treatment of non-linear coherent states, we are able to derive an explicit expression resembling Eq.~\eqref{eqn:genericDephasingChannel} for this deformed dephasing channel. This approach allows us to study the impact of a non-linear environment. Surprisingly, we observe that the non-linearity $\lambda$ is able to compensate the dephasing rate in the case of negative values of the non-linearity parameter $\lambda$. In addition, we indicate that the quantum capacity of the deformed dephasing bosonic channel is strictly decreasing as a function of $\lambda$. 

Note that for the deformed channel in Eq.~\eqref{eqn:kerr-unitary}, the ``free'' evolution of the system mode $S$ now includes a quadratic energy dependence and so the phase of $S$ evolves non-linearly with its photon number (which nonetheless is preserved) and also non-linearly with the environmental mode operators.

The paper is organized as follows.
First, we define our model in Section \ref{sec:system-model}. Then, we state our main results in Section \ref{sec:results}. Finally, in Section \ref{sec:methods} we present the methods and details regarding the proof (of Section \ref{subsec:proof-of-main-theorem}) and we give some further, numerical methods (Subsection \ref{subsec:numerical-methods}) of our work.

\section{ Deformed quantum dephasing channel\label{sec:system-model}}

Let us consider the Hamiltonian of an anharmonic oscillator  
\begin{eqnarray}\label{qcd1}
\hat{H}=  \Omega \hat{a}^{\dagger}\hat{a} +  \frac{\lambda}{2} \hat{a}^{\dagger 2}\hat{a}^{2},
\end{eqnarray}
in which $\hat{a}^{\dagger}$ and $\hat{a}$ are the creation and annihilation bosonic operators, $\hat{n}=\hat{a}^{\dagger}\hat{a}$ is a number operator and $\frac{\lambda}{2}$, the so-called anharmonicity,  is related to the non-linear susceptibility of the Kerr medium.
The Hamiltonian (\ref{qcd1}) has been vastly applied to model  different phenomena. It can be mathematically considered as a description of a position-dependent quantum oscillator \cite{33}; it models an oscillator  confined in a finite or infinite well \cite{34};  it  describes a confined oscillator on a one-dimensional space with constant curvature, i.e., circle and hyperbolic \cite{35}. In addition, it models a Kerr medium \cite{36} and a Transmon gate \cite{37}. \\
\indent 
We define deformed annihilation and creation operators as
\begin{eqnarray}\label{eqgb2}
\hat{A}(\lambda,\omega)=\hat{a}f(\hat{n})=f(\hat{n}+\mathbbm{1})\hat{a},\\
\hat{A}^{\dagger}(\lambda,\omega)=f(\hat{n})\hat{a}^{\dagger}=a^{\dagger}f(\hat{n}+\mathbbm{1}),
\end{eqnarray}
in which  $\omega:=\Omega-\tfrac{\lambda}{2}$ and the deformation function is given by
\begin{eqnarray}
f(\hat{n})=\sqrt{\mathbbm{1}+\tfrac{\lambda}{2\omega}\hat{n}}.
\end{eqnarray}

 \begin{figure*}[]
    \centering
    \includegraphics[width=16cm]{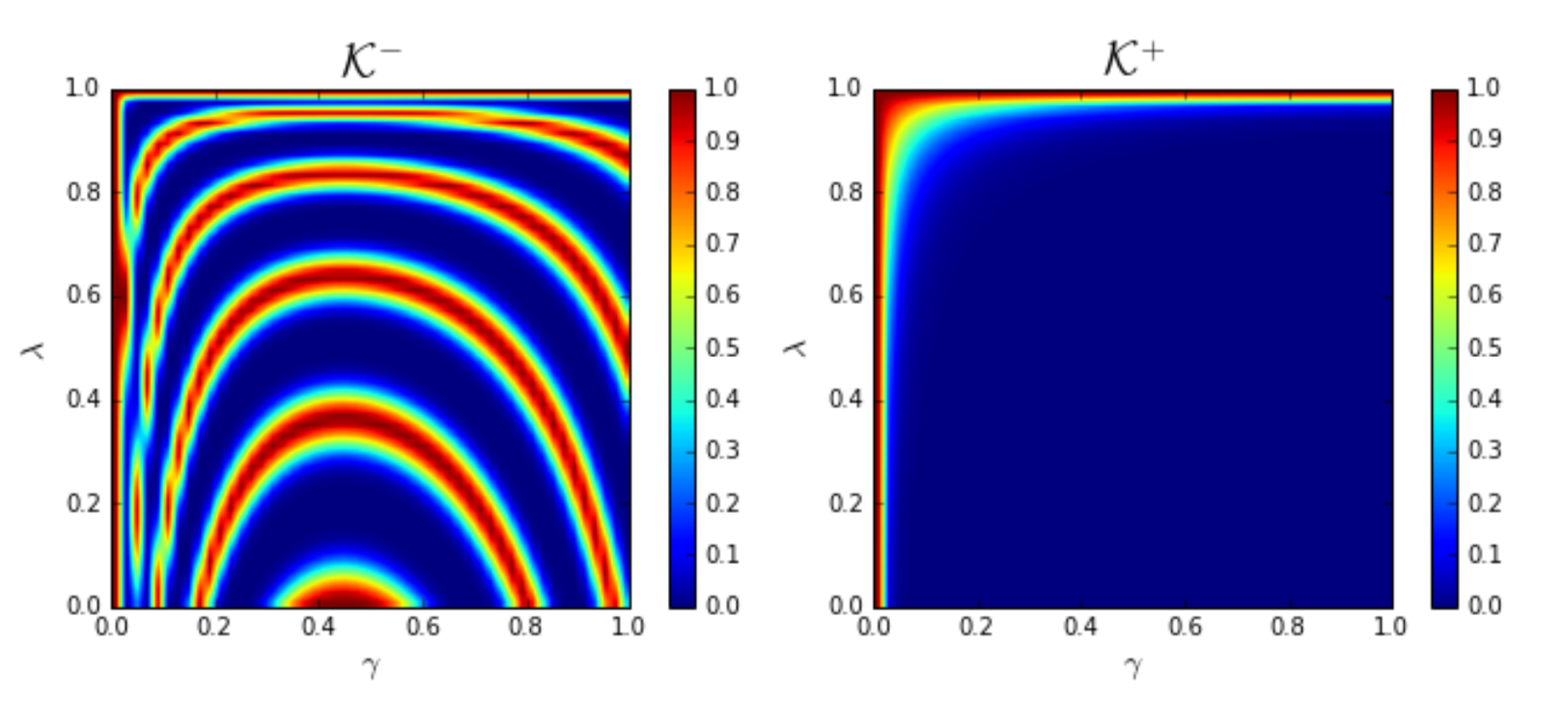}
    \caption{Density plot of functions $\mathcal{K}^{-}(\tau_{n},\tau_{m};\lambda)$ and $\mathcal{K}^{+}(\tau_{n},\tau_{m};\lambda)$ as a function of parameters $\lambda$ and $\gamma$. }
    \label{fig1-2}
\end{figure*}
\begin{figure*}[t]
    \centering
    \includegraphics[width=16cm]{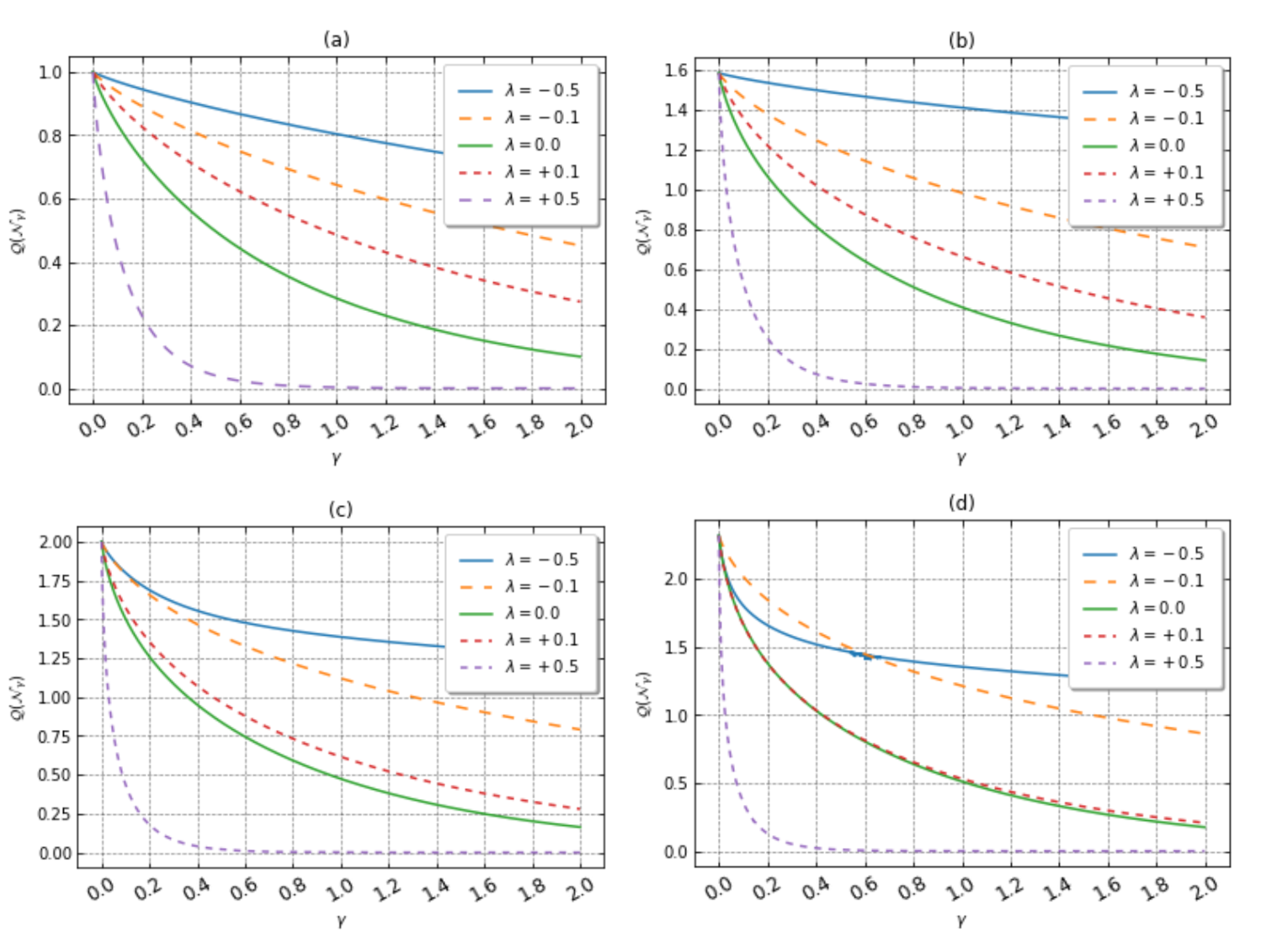}
    \caption{The optimal von Neumann entropy versus $\gamma$, for different values of $\lambda$, for $N=1, 2, 3$. and $N=4$, in plots (a), (b), (c) and (d), respectively.  }
    \label{fig1}
\end{figure*}
With this choice, we can rewrite the Hamiltonian (\ref{qcd1}) as
\begin{eqnarray}\label{qcd2}
\hat{H}=\omega\hat{A}^{\dagger}(\lambda,\omega)  \hat{A}(\lambda,\omega). 
\end{eqnarray}
In what follows, we assume $\lambda$ and $\omega$ to be arbitrary but fixed, and for simplicity, we write $A(\lambda,\omega)=\hat{A}$. Since these deformed operators are assumed to dictate the behaviour both on system and environment, we equivalently write $\hat B$ for the deformed annihilation operator on the environment. 
Now, similar to the definition of the dephasing channel \cite{25,28}, we can define a deformed quantum dephasing channel. 
\begin{definition}\label{def:dephasing-channel}
    Given $\lambda\in\mathbbm R$ and $\gamma,\omega>0$ the deformed quantum dephasing channel is defined as \begin{eqnarray}\label{gbqc22}
    \mathcal{N}_{\gamma}(\rho):=\Tr_{E} \left[U\left(\hat{\rho}\otimes \hat{\sigma}\right)U^{\dagger}\right]
    \end{eqnarray}
    in which the unitary operator $\hat{U}$ is 
    \begin{eqnarray}\label{eqgbqc6}
    \hat U=\exp\left[-i\sqrt{\gamma}\hat{A}^{\dagger}\hat{A}\left(\hat{B}+\hat{B}^{\dagger}\right)\right]\in\mathcal L(\mathcal F\otimes\mathcal F)
    \end{eqnarray}
    where $\hat{A}$ and $\hat{B}$ are the deformed annihilation operators, defined by Eq.~(\ref{eqgb2}), of the system and environment, respectively.
\end{definition}

\section{Results\label{sec:results}}
Based on Definition \ref{def:dephasing-channel}, we can provide as our main result an analytic expression of the action of the dephasing channel. Since we require the Hamiltonian \eqref{qcd1} to have only non-negative energies (eigenvalues), the Hilbert space representing the quantum system has be of finite dimension if $\lambda<0$ holds. Therefore, our result is split into the two different cases when $\lambda\geq0$ and when $\lambda<0$. 
Then, we can write our main theorem as follows:
\begin{theorem}\label{theorem:main}
Let $\rho=\sum_{m,n=0}^\infty\ketbra{m}{n}\in\mathcal S(\mathcal F)$. Let $p=(\gamma,\lambda,\omega)$. Then the deformed dephasing bosonic channel is given by
\begin{eqnarray}\label{eqn:main-result}
\mathcal{N}_{\gamma}(\rho)
    =\sum_{n,m=0}^{\infty}
 \mathcal{K}^{p}_{n,m} \rho_{nm}\ket{n}\bra{m}
\end{eqnarray}
in which $\mathcal{K}^p_{n,m}$ has a different structure depending on $\lambda$ as follows: if $\lambda>0$, we set $\tau_n:=\sqrt{\gamma \tfrac{\lambda}{2\omega}}\ n(1+n\tfrac{\lambda}{2\omega})$, then with $\nu:=1+2\omega/\lambda$ we have
\begin{align}\label{eqgbv312}
   \mathcal{K}^{p}_{m,n}=\frac{\left(1-\tanh^{2}\tau^{
 +}_{n}\right)^{\nu}\left(1-\tanh^{2} \tau^{
 +}_{m}\right)^{\nu}}{\left[1-\tanh \tau_{n}^{
 +}\tanh \tau^{
 +}_{m}\right]^{2\nu}}.
\end{align}
If $\lambda=0$, then 
\begin{align}\label{eqgbv315}
    \mathcal K^p_{n,m} = e^{-\frac{\gamma}{2}(m-n)^{2}}.
\end{align}
Lastly if $\lambda<0$, then 
\begin{align} \label{eqgbv314}
    \mathcal{K}^{p}_{m,n}=\frac{(1+\tan\tau_{m}\tan\tau_{n})^{2\nu}}{(1+\tan^{2}\tau_{m})^{\nu}(1+\tan^{2}\tau_{n})^{\nu}}
\end{align}
in which $\nu:=1+2\omega/\lambda$. 
\begin{proof}[Sketch of the Proof:]
We give a high-level sketch of the proof here, for details, see Sec.~\ref{subsec:proof-of-main-theorem}. We start out by calculating the action of the channel according to Definition \ref{def:dephasing-channel}:
\begin{align}
\mathcal{N}_{\gamma}(\rho)
    &=\Tr_{E}\left[\hat U\rho\otimes\ket{0}\bra{0}\hat U^\dag\right]\\
    &=\sum_{m,n}\rho_{m,n}\Tr_{E}\left[U\ketbra{m}{n}\otimes\ketbra{0}{0}U^\dagger\right]\nonumber\\
    &=\sum_{m,n=0} \rho_{m,n}\ketbra{m}{n}\Tr_{E}\left[\ketbra{-\mathbbm{i}\tau_{m};\lambda}{-\mathbbm{i}\tau_{n};\lambda}\right],\nonumber
\end{align}
in which we employed the fact that $\hat A^\dag\hat A(\hat B+\hat B^\dag)$ is a simple tensor product over system and bath, and  $\ket{-\mathbbm{i}\tau_{m};\lambda}$ denotes the associated coherent states. The calculation of the partial trace over the environment is thus equivalent to calculating the scalar products
\begin{eqnarray}\label{gbqc20}
\mathcal{K}^{p}_{n,m}=\braket{-\mathbbm{i}\tau_{n};\lambda}{-\mathbbm{i}\tau_{m};\lambda}
\end{eqnarray}
of non-linear coherent states. The details of the proof are the content of Sec.~\ref{subsec:proof-of-main-theorem}, the technical lemmata are to be found in Sec.~\ref{sec:methods}.
\end{proof}
\end{theorem}
Note that in the case of $\lambda\geq 0$, the coefficient $\mathcal{K}^{p}$, $p=+,0$, in which $\lambda$ is equal or grater than zero, approaches zero exponentially when $m\neq n$, when $\gamma \rightarrow \infty$, which means the off-diagonal elements of the density matrix map to zero in this channel, for $\gamma\gg 1$; in other words, the deformed dephasing bosonic channel with $\lambda\geq 0$ causes a decohecnce process to occur, while in the    case $\lambda<0$, as a periodic nature of the function $\mathcal{K}^{-}$,  we are able to suppress the decoherence process.   Fig.~\ref{fig1-2} illustrates the overlap of two non-linear coherent states, i.e., the relation (\ref{gbqc20}), as a function of parameters $\lambda$ and $\gamma$. 
 Especially, in the case $\lambda<0$, by adjusting the parameter $\lambda$, the off-diagonal elements can be preserved as well. It is therefore interesting to investigate the properties of this channel for the purpose of entanglement transmission and show how to calculate its quantum capacity:
\begin{definition}
 The quantum capacity of the bosonic dephasing channel is defined as
\begin{eqnarray}\label{gbqc32}
\mathcal{Q}(\mathcal{N}_{\gamma} ) = \max_{\hat\rho} J(\hat{\rho}, \mathcal{N}_{\gamma}),
\end{eqnarray}
where
\begin{eqnarray}
 J(\hat{\rho}, \mathcal{N}_{\gamma})=S(\mathcal{N}(\hat{\rho}))-S(\mathcal{N}^{c}(\hat{\rho}))
\end{eqnarray}
and $S(\hat{\rho})=-\Tr \left[\hat{\rho}\log_{2}\hat{\rho}\right]$ is the von Neumann entropy and  
the complementary channel $\mathcal{N}^{c}_{\gamma}$ is given by
\begin{eqnarray}\label{def:complementary-dephasing-channel}
\mathcal{N}^{c}(\hat\rho)=\Tr_{S} \left[U\left(\hat{\rho}\otimes \hat{\sigma}\right)U^{\dagger}\right].
\end{eqnarray}
\end{definition}
We show in Lemma~\ref{lem:optimal-input} that the optimal input states in the above definition are diagonal in the number state basis.

Therefore, the maximization in the relations  (\ref{gbqc32}) leads to the maximization over a classical probability distribution:
\begin{eqnarray}\label{gbqc37}
\mathcal{Q}(\mathcal{N}_{\gamma})&=& \max_{p_{n}}\Bigg[S\left(\sum_{n=0}^{N}P_{n}
\ket{n}\bra{n}\right)\nonumber\\
&-&
S\left(\sum_{n=0}^{N}P_{n}
\ket{i\sqrt{\gamma} n,\lambda}\bra{i\sqrt{\gamma} n,\lambda}\right)
\Bigg],
\end{eqnarray}
for the following inputs which are given in a truncated Hilbert space $\rho=\sum_{n=0}^{N}P_n\ketbra{n}{n}$, with the finite   average energy, i.e., $\sum_{n=0}^{N}P_{n}\varepsilon^{\pm}(n)\leq E$ with
\begin{eqnarray}
\varepsilon^{\pm}(n)=n\pm \frac{|\lambda|}{2}n^{2}.
\end{eqnarray}
Note that for $\lambda<0$, we should also impose the following condition, $N\leq d=\lfloor 2\nu \rfloor$. We evaluate these capacities numerically for $N=1, \cdots, 4$. \\
\indent  For $N=1$,    we consider the following  input state:
\begin{eqnarray}
\Omega=p_{1}\ket{n}\bra{n}+p_{2}\ket{n+m}\bra{n+m}
\end{eqnarray}
where $n$, and $m$ are arbitrary non-negative integers. Plot (a)-Fig. \ref{fig1} illustrates $\mathcal{Q}(\mathcal{N}_{\gamma})$, i.e., the relation (\ref{gbqc37}), as a function of the dephasing parameter $\gamma$, for different values  $\lambda$. The plot indicates that for the linear environment, $\lambda=0$, and the Kerr medium, with $\lambda>0$, the quantum capacity decreases by increasing the dephasing parameter. 
Moreover, the quantum capacity for the Kerr medium with $\lambda<0$ is greater than for the linear case and for the Kerr medium with positive non-linearity. 
 However, for every $\lambda<0$, the simulation results indicate that the capacity is a periodic function of $\gamma$, as illustrated in Fig.~\ref{fig3}.
\begin{theorem}\label{theorem:capacities}
Let $\hat\rho=\ket{\alpha}\bra{\alpha}\otimes \ket{0}_E\bra{0}$, which $\ket{\alpha}$ is a coherent state, i.e., 
\begin{eqnarray}\label{eqgbv323}
\ket{\alpha}=e^{-|\alpha|^{2}/2}\sum_{n=0}^{\infty}\frac{\alpha^{n}}{n!}\ket{n}.
\end{eqnarray}
Then, for $\lambda<0$, we get
\begin{eqnarray}
\mathcal{N}_{\gamma}=e^{-|\alpha|^{2}}\sum_{n,m=0}^{2j}\mathcal{K}^{-}_{n,m}\frac{\alpha^{n}\alpha^{\ast m}}{\sqrt{n!m!}}\ketbra{n}{m}
\end{eqnarray}
in which $\mathcal{K}^{-}_{n,m}$ is given by the relation (\ref{eqgbv314}).
For $\lambda=0$, we obtain
\begin{eqnarray}
\mathcal{N}_{\gamma}=e^{-|\alpha|^{2}}\sum_{n,m=0}^{\infty}e^{-\gamma(n-m)^{2}}\frac{\alpha^{n}\alpha^{\ast m}}{\sqrt{n!m!}}\ketbra{n}{m}.
\end{eqnarray}
Finally, for $\lambda>0$, we achieve
\begin{eqnarray}
\mathcal{N}_{\gamma}=e^{-|\alpha|^{2}}\sum_{n,m=0}^{\infty}\mathcal{K}^{+}_{n,m}\frac{\alpha^{n}\alpha^{\ast m}}{\sqrt{n!m!}}\ketbra{n}{m}
\end{eqnarray}
in which $\mathcal{K}^{+}_{n,m}$ is given by the relation (\ref{eqgbv312}).
\begin{proof}
By inserting the relation (\ref{eqgbv323}) into the definition (\ref{eqn:main-result}), we can directly obtain the result.  
\end{proof}
\end{theorem}
\section{ Proof of Theorem \ref{theorem:main}\label{subsec:proof-of-main-theorem}}
We calculate the partial trace over the environment as 
\begin{align}
\mathcal{N}_{\gamma}(\rho)\label{eqn:partial-trace}
    &=\Tr_{E}\left[\hat U\rho\otimes\ket{0}\bra{0}\hat U^\dag\right]\\
    &=\sum_{m,n}\rho_{m,n}\Tr_{E}\left[U\ketbra{m}{n}\otimes\ketbra{0}{0}U^\dagger\right],
\end{align}
where the sums can run from $0$ to $\infty$ in case of $\lambda\geq 0$, or to a finite number $d$ in case of $\lambda<0$, i.e., a finite dimension $d$. 
The calculation of the partial trace in \eqref{eqn:partial-trace} relies on the following:
\begin{align}
    \hat U\ket{m}\ket{0} 
       &=\sum_{k=0}^\infty(-\mathbbm{i}\sqrt{\gamma}A^\dag A)^k\frac{(B+B^\dag)^k}{k!}\ket{m}\ket{0}\\
       &=\ket{m}\sum_{k=0}^\infty\frac{(-\mathbbm{i}\tau_m(B+B^\dag))^k}{k!}\ket{0}
\end{align}
where $\tau_{m}=\sqrt{\gamma}\cdot\bra{m}\hat A^\dag\hat A\ket{m}$ is defined explicitly in Theorem \ref{theorem:main}. 
\begin{figure}[t!]
    \centering
    \includegraphics[width=9cm]{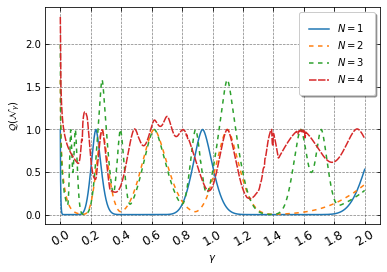}
    \caption{The optimal von Neumann entropy versus $\gamma$, for $N=1, 2, 3$. and $N=4$, with  $\lambda=-2$.  }
    \label{fig3}
\end{figure}
We can thus define the deformed displacement operator as
\begin{eqnarray}
D(\mu,p)&=&
\sum^{\infty}_{k=0}
\frac{\left[-\mathbbm{i} 
\mu(\hat{B}+\hat{B}^\dagger)
\right]^{k}}{k!}\nonumber\\
    &=&\exp\left[-\mathbbm{i} 
    \mu(\hat B+\hat B^\dag)\right],
\end{eqnarray}
to arrive at 
\begin{align}
    \hat U\ket{m}\ket{0} = \ket{m}D(\sqrt{\gamma}\tau_m,p)\ket{0}.
\end{align}

This shows that calculation of quantities like $\exp\left[-\mathbbm{i}\mu(\hat B+\hat B^\dag)\right]\ket{0}$ is relevant to proving Theorem \ref{theorem:main}. 
To compute matrix exponentials of the form $\exp[-\mathbbm{i}\mu(\hat B+\hat B^\dag)]$, we use the method derived in  Ref.~\cite{38} and the commutation relations (\ref{eqgbqc31})-(\ref{eqgbqc33}) which yield our Lemma \ref{lem:decomposition}. 
The deformed annihilation and creation operators and the Hamiltonian $\hat H$ take the form $\{\hat A,\hat{K}_{0},\hat A^\dag\}$ of a Holstein-Primakoff representation of a $\lambda$-deformed algebra.
Note that in the special case in which $\lambda=\pm2$, it can be identified as a $su(1,1)$ and a $su(2)$-algebra, respectively. \\
\indent For simplicity, we derive the Gaussian decomposition of the displacement operators separately for the positive and negative sign of $\lambda$.  
\subsubsection{The case $\lambda>0$}
Let us assume for a start that $\hat A = \sqrt{x+y\hat n}\hat a$ and $\hat B = \sqrt{x+y\hat n}\hat b$ with generic parameters $x,y$ to be adjusted later. Using Lemma \ref{lemma:exp(b+bdag)-applied-to-zero} we can develop a more explicit expression for the generalized coherent state. To eventually arrive at an analytic expression for the action of the channel $\mathcal N_\gamma$, we calculate partial traces by relying on the following relation:
\begin{align}
     U\ket{m}\otimes\ket{0}&= \ket{m}\otimes\exp\left[-\mathbbm{i}\mu_{m}(\hat B+\hat B^\dag)\right]\ket{0}
\end{align}
in which we set $\mu_{m}:=\sqrt{\gamma}\bra{m}\hat A^\dag\hat A\ket{m}$ leading to $\mu_m=\sqrt{\gamma}m(x+y\cdot m)$; thus, we obtain for the parameter $t$ in Lemma \ref{lemma:exp(b+bdag)-applied-to-zero} $t_m:=\tanh\left( y^{1/2}\tau_m\right)$ and via Lemma \ref{lemma:exp(b+bdag)-applied-to-zero} the relation
\begin{align}
    \tr_{E}&(U\ketbra{m}{n}\otimes\ketbra{0}{0}U^\dag)\nonumber\\
        &=\ket{m}\bra{n}tr(e^{-\mathbbm{i}\mu_m(\hat B+\hat B^\dag)}\ket{0}\bra{0}e^{\mathbbm{i}\mu_m(\hat B+\hat B^\dag)})\\
        &=\ket{m}\bra{n}c_mc_n\langle\phi^{(x+y)/y}_{t_m},\phi^{(x+y)/y}_{t_n}\rangle\\
        &=\ket{m}\bra{n}c_mc_n(1-t_mt_n)^{-(x+y)/y},
\end{align}
where
\begin{align}
    c_m&=\cosh^{2\tfrac{x+y}{-y}}(\sqrt{y}\mu_m)=\left(1-\tanh^{2}(\sqrt{y}\mu_m)\right)^{\tfrac{x+y}{y}}\\
    t_m&=\tanh(\sqrt{y}\mu_m).
\end{align} 
Thus setting 
\begin{align}
    \kappa_m=\sqrt{y\cdot\gamma}m(x+y\cdot m)
\end{align}
we get the result
\begin{align}
    &\tr_{E}(U\ketbra{m}{n}\otimes\ketbra{0}{0}U^\dag)\nonumber\\
        &\ =\frac{((1-\tanh^{2}\kappa_m)(1-\tanh^{2}\kappa_n))^{\tfrac{x+y}{y}}}{(1-\tanh\kappa_m\tanh\kappa_n)^{(x+y)/y}}\ket{m}\bra{n}.
\end{align}
We observe that a transformation $x':=\theta x$, $y':=\theta y$, $\gamma'=\theta^{-3}$ with $\theta>0$ yields
\begin{align}
    \tfrac{x'+y'}{y'}=\tfrac{x+y}{y},\\
    \kappa'_m=\sqrt{y'\gamma'}m(x'+y'm)=\kappa_m.
\end{align}
By our convention, we have $\omega:=\Omega-\lambda/2$, $x=1$ and $y=\tfrac{\lambda}{2\omega}$. This results in $\tau_m=\sqrt{\tfrac{\gamma\lambda}{2\omega}}m(1+\tfrac{m\lambda}{2\omega})$. 
\subsubsection{The Case $\lambda<0$}
By using Lemma \ref{lem:decomposition}, we can write:
\begin{lemma}\label{lemma:exp(b+bdag)-applied-to-zero1}
    Let $\mu\in\mathbbm{C}$ and $\hat B=\hat B_\lambda$ be as defined in \eqref{eqgb2}. Then it holds that
    \begin{eqnarray}
        e^{\mu(\hat B+\hat B^\dag)}\ket{0} &=& e^{|\lambda|\nu/2\ln(\zeta_{0})}\sum_{k=0}^{\lfloor 2\nu\rfloor}\left(\zeta_{+}\sqrt{\frac{|\lambda|}{2}}\right)^{k}\sqrt{\frac{(2\nu)_{k}}{k!}}\ket{k}\nonumber\\
        &=& \cos^{2\nu}\left(\sqrt{\frac{\lambda}{2}}|\mu| \right)\sum_{k=0}^{\lfloor 2\nu\rfloor}e^{-im\phi} \nonumber\\
       &\times& \tan^{k}\left(\sqrt{\frac{\lambda}{2}}| \mu | \right) \sqrt{\frac{(2\nu)_{k}}{k!}}\ket{k}
    \end{eqnarray}
    where $\nu:=|\lambda|^{-1}-1/2$,  $(x)_{q}=\Gamma(x+1)/\Gamma(x-q+1)$ is the  falling factorial and $\lfloor 2\nu\rfloor$ is the floor function defined by  $\text{floor}(x)=\lfloor x\rfloor$, i.e., the greatest integer less than or equal to $x$.
\end{lemma}
Using Lemma \ref{lemma:exp(b+bdag)-applied-to-zero1} we can develop a more explicit expression for the generalized coherent state. To eventually arrive at an analytic expression for the action of the channel $\mathcal N_\gamma$, we calculate partial traces by relying on the following:
\begin{lemma}\label{lem:product-of-deformed-coherent-states1}
    Let, for generic complex numbers $x$ and fixed $N$ states, $\phi_x$ be defined via
    \begin{align}
        \phi_x=\sum_{k=0}^{\lfloor \alpha\rfloor} \sqrt{\frac{(\alpha)_{k}}{k!}}x^{k}\ket{k}. 
    \end{align}
    Then
    \begin{align}
        \langle \phi_x,\phi_y\rangle
            &= (1+x\cdot y)^{\alpha}.
    \end{align}
    \begin{proof}
    Using the generalized  binomial formula the above sum can be turned  into the result \cite{licciardiguide}.
    \end{proof}
\end{lemma}

Using Lemma \ref{lemma:exp(b+bdag)-applied-to-zero1} and Lemma \ref{lem:product-of-deformed-coherent-states1} we thus conclude that 
\begin{align}
     U\ket{m}\otimes\ket{0}&= \ket{m}\otimes\exp\left[\mu_{m}(\hat B+\hat B^\dag)\right]\ket{0},
\end{align}
in which we set $\mu_{m}:=-\mathbbm{i}\tau_{m}$;  thus we obtain via Lemma \ref{lemma:exp(b+bdag)-applied-to-zero1} the relation
\begin{align}
    \tr_{E}(U&\ketbra{m}{n}\otimes\ketbra{0}{0}U^\dag)=\nonumber\\
    &=(1+x_{m}^{2})^{-\nu}(1+y_{n}^{2})^{-\nu}\langle\phi_x,\phi_y\rangle \ketbra{m}{n}\\
    &=
    \frac{(1+x_{m}y_{n})^{2\nu}}{(1+x_{m}^{2})^{\nu}(1+y_{n}^{2})^{\nu}}\ketbra{m}{n},
\end{align}
with $x_m:=\tan(\tau_m)$. Therefore for $\lambda<0$ we have
\begin{eqnarray}
\mathcal{K}^{p}_{m,n}=\frac{(1+\tan\tau_{m}\tan\tau_{n})^{2\nu}}{(1+\tan^{2}\tau_{m})^{\nu}(1+\tan^{2}\tau_{n})^{\nu}}.
\end{eqnarray}

\section{Methods}\label{sec:methods}
\subsection{Properties of Deformed Algebra}
Before we go into detailed calculations, we first derive some important commutation relations in full generality, which we summarize in the Lemma below.
\begin{lemma}\label{lem:properties-of-A}
    Let $\{\ket{n}\}_{n=0}^d$ be a countable orthonormal basis of a Hilbert space, where $d\in\mathbb N$ or $d=\infty$. Let $\hat A:=\hat a \hat g(x,y)$ with $\hat g(x,y):=\sqrt{x+y\hat a^\dag\hat a}$ where $x\geq0$ and $y\geq-x/d$ and $\hat a^\dag,\hat a$ are the creation- and annihilation operators satisfying $\hat a\ket{n}=\sqrt{n}\ket{n-1}$, $\hat a^\dag\ket{n}=\sqrt{n+1}\ket{n+1}$ and $[\hat a,\hat a^\dag]=\mathbbm{1}$. Define $\hat K_0:=[\hat A,\hat A^\dag]/2$ and $\hat n:=\hat a^\dag\hat a$. Then the following are true:
    \begin{enumerate}
       \item $\hat g(x,y)\geq0$
       \item $a\geq0$, $b\geq-x/d$ implies $[\hat g(x,y),\hat g(a,b)]=0$
       \item $\hat A^\dag \hat A  = (x+y)\hat n + y\hat a^{\dag2}\hat a^2$
       \item  $[\hat K_0,\hat A] =  -y\hat A$
       \item $[\hat K_0,\hat A^\dag] = y\hat A^\dag$.
    \end{enumerate}
\end{lemma}
Lemma \ref{lem:properties-of-A} allows us to prove further statements which enter the proof of Theorem \ref{theorem:main}. The second tool we need is the Gaussian decomposition of $\exp\left[\beta A^{\dagger}-\beta^{\ast}A\right]$, which is derived in the following Lemma:
\begin{lemma}\label{lem:decomposition}
    Let the operators $\hat A$, $\hat A^\dag$ and $\hat K_0$ satisfy the commutation relations
    \begin{eqnarray}
    [\hat{A},\hat{A}^{\dagger}]&=&2\hat{K}_{0},\label{eqgbqc31}\\ \newline 
    [\hat{K}_{0},\hat{A}]&=&-\frac{\lambda}{2}\hat{A},\label{eqgbqc32}\\ \newline 
    [\hat{K}_{0},\hat{A}^{\dagger}]&=&\frac{\lambda}{2}\hat{A}^{\dagger}\label{eqgbqc33}.
    \end{eqnarray}
    Then the Gaussian decomposition of  the  displacement operator $D(\beta,\lambda)=\exp\left[\beta A^{\dagger}-\beta^{\ast}A\right]$ is given by
    \begin{eqnarray}\label{gbqca7}
    D(\beta,\lambda^{+})=e^{\beta A^{\dagger}-\beta^{\ast}A}=e^{\zeta A^{\dagger}}e^{\ln[\zeta_{0}] K_{0}} e^{-\zeta^{\ast}A}
    \end{eqnarray}
    where
    \begin{eqnarray}
    \zeta &=&\tfrac{\beta}{|\beta|}
    \sqrt{\frac{2}{\lambda}} \tanh\left(\sqrt{\frac{\lambda}{2}}|\beta|\right),
    \label{def:zeta}
    \\
    \zeta_{0} &=& \cosh^{-4/\lambda}\left(\sqrt{\frac{\lambda}{2}}|\beta|\right).\label{eq16}
    \end{eqnarray}
\end{lemma}

Via application of the Gaussian decomposition from Lemma \ref{lem:decomposition}, a substantial part of the proof of Theorem \ref{theorem:main} can be reduced to the following Lemma:
\begin{lemma}\label{lemma:exp(b+bdag)-applied-to-zero}
    Let $\mu\in\mathbbm{C}$ and $\hat B=\hat{a}(x\mathbbm{1}+y\hat n)^{1/2} $. Let $\hat K_0:=[\hat B,\hat B^\dag]/2$. Then 
\begin{align*}
     e^{-\mathbbm{i}\mu(\hat B+\hat B^\dag)}&\ket{0}
        = \cosh^{-2\tfrac{x+y}{y}}(\sqrt{y}|\mu|))\sum_{k=0}^\infty t^k\tfrac{\sqrt{(x/y)^{(k)}}}{\sqrt{k!}}\ket{k}
\end{align*}
    where $t:=\tfrac{\beta}{|\beta|}\tanh\left( y^{1/2}|\beta|\right)$ and the symbol $(x)^{(q)}:=\Gamma(x+q)/\Gamma(x)$ denotes the rising factorial.
\end{lemma}
Finally, Umbral calculus \cite{40} is employed to calculate sums of the following type:
\begin{lemma}\label{lem:product-of-deformed-coherent-states}
    Let, for generic complex numbers $x,y$ and a real number $\alpha$, vectors $\phi_x^\alpha $ be defined via
    \begin{align}
        \phi_x^\alpha=\sum_{k=0}^\infty \sqrt{\frac{(\alpha)^{(k)}}{k!}}x^{k}\ket{k}. 
    \end{align}
    Then
    \begin{align}
        \langle \phi_x^\alpha,\phi_y^\alpha\rangle
            &= (1-x\cdot y)^{-\alpha}.
    \end{align}
\end{lemma}
With Lemmas \ref{lem:properties-of-A}, \ref{lem:decomposition},  \ref{lemma:exp(b+bdag)-applied-to-zero} and \ref{lem:product-of-deformed-coherent-states} at hand we are ready for the proof of Theorem \ref{theorem:main}.
\subsubsection{Kraus Representation}
The Kraus representation of the channel is given by
\begin{eqnarray}
\rho \mapsto \mathcal{N}_{\gamma}(\rho),
\end{eqnarray}
in which 
\begin{eqnarray}
\mathcal{N}_{\gamma}(\rho)
    &=&\Tr_{E}\left[\hat U\rho\otimes\ket{0}_{E}\bra{0}\hat U^\dag\right]\nonumber\\
    &=&\sum_{l=0}^{\infty}\hat{K}_{l}\rho_{S} \hat{K}_{l}^{\dagger}. 
\end{eqnarray}
Now by considering the fact that $\sum_{l=0}^{\infty}\ket{l}_{E}\bra{l}=\mathbb{I}$ and using  the relation (\ref{eqn:main-result}), we can obtain
 \begin{eqnarray}
 K_{l}(\lambda^{-})&=& _{E}\bra{l}U\ket{0}_{E}\nonumber\\
 &=&\cos^{2\nu}\left(\sqrt{\frac{|\lambda|}{2}} \hat{A}^{\dagger}\hat{A}\right)\nonumber\\
&\times & \sqrt{\frac{(2\nu)_{l}}{l!}}  \tan^{l}\left(\sqrt{\frac{|\lambda|}{2}} \hat{A}^{\dagger}\hat{A}\right),
 \end{eqnarray}
and 
 \begin{eqnarray}
 K_{l}(\lambda^{+})&=& _{E}\bra{l}U\ket{0}_{E}\nonumber\\
 &=& \cosh^{-2\nu}\left(\sqrt{\frac{|\lambda|}{2}} \hat{A}^{\dagger}\hat{A}\right)\nonumber\\
 &\times&\sqrt{\frac{(2\nu)^{(l)}}{l!}}  \tanh^{l}\left(\sqrt{\frac{|\lambda|}{2}} \hat{A}^{\dagger}\hat{A}\right).
 \end{eqnarray}
 
\subsection{Quantum  Capacity}

For the channel $\mathcal{N}_{\gamma}$, the complementary channel $\mathcal{N}^{c}_{\gamma}$ is defined as
\begin{eqnarray}\label{gb22}
\mathcal{N}_{\gamma}^{c}(\rho)=\Tr_{S} \left[U\left(\hat{\rho}\otimes \hat{\sigma}\right)U^{\dagger}\right].
\end{eqnarray}
By imposing the above-mentioned condition, i.e.,   attributing a ground state to the environment,  the complementary channel is  given by
\begin{eqnarray}\label{comp}
\mathcal{N}^{c}(\hat{\rho})&=&\Tr_{S}\left[\rho_{n,m}\ket{n}\bra{m}\otimes \ket{-i\tau_{n},\lambda^{\pm}}\bra{-i\tau_{m},\lambda^{\pm}}\right]\nonumber\\
&=&\sum_{n}\rho_{n,n} \ket{-i\tau_{n},\lambda^{\pm}}\bra{-i\tau_{n},\lambda^{\pm}},
\end{eqnarray}
which  is a mixture of deformed coherent states. Note that by following the same method as in Ref.~\cite{arqand2020quantum}, it is possible to show that the complementary channel $\mathcal{N}_{\gamma}^{c}$  is entanglement breaking. In fact, by considering the following state
\begin{eqnarray}
\ket{\Psi}=\sum_{n}\lambda^{n}
\ket{n}_{R}\ket{n}_{E},
\end{eqnarray}
in which $\lambda$ is squeezing parameter, $\ket{n}_{R}$ and $\ket{n}_{E}$ are, respectively,  a state of the Hilbert space of a reference state and a state of the environment. Therefore, we can define the following channel:
\begin{eqnarray}
(\mathbb{I}_{R}\otimes \mathcal{N}_{\gamma}^{c}) \ket{\Psi}\bra{\Psi}=\sum_{n}\lambda^{2n}\ket{n}\bra{n}\otimes \ket{-i\tau_{n};\lambda} \bra{-i\tau_{n};\lambda},\nonumber\\
\end{eqnarray}
which is a mixture of product states and hence is a separable state for any value of $\lambda$, where $\mathcal{N}_{\gamma}^{c}$ is entanglement breaking and $\mathcal{N}_{\gamma}$ is degradable.
 Therefore,  the quantum capacity of the bosonic dephasing channel is given by
\begin{eqnarray}\label{gb32}
\mathcal{Q}(\mathcal{N}_{\gamma} ) = \max_{\rho} J(\hat{\rho}, \mathcal{N}_{\gamma}).
\end{eqnarray}
where
\begin{eqnarray}\label{eq68}
 J(\hat{\rho}, \mathcal{N}_{\gamma})=S(\mathcal{N}(\hat{\rho}))-S(\mathcal{N}^{c}(\hat{\rho}))
\end{eqnarray}
where $S(\hat{\rho})=-\Tr \left[\hat{\rho}\log_{2}\hat{\rho}\right]$ is the von Neumann entropy.

\begin{lemma}\label{lemma:covariant}
The bosonic dephasing channel (\ref{gbqc22}) is a   phase-covariant channel.
\end{lemma}
\begin{proof}
Using the explicit definition of the deformation function $f(\hat{n})=\sqrt{\mathbbm{1}+\tfrac{\lambda}{2\omega}\hat{n}}$, we have 
\begin{eqnarray}\label{eq699}
K_{0}=\frac{1}{2}[A,A^{\dag}]=\frac{\lambda}{2\omega}n+\frac{1}{2}(1+\frac{\lambda}{2\omega}),
\end{eqnarray}
which is a function of number operator $n$.  By defining the unitary operator $\hat{U}$ as
\begin{eqnarray}
\hat{U}_{\theta}=e^{i\hat{K}_{0}\theta},  \ \theta \in [0,2\pi),
\end{eqnarray}
we see that $\hat U_{\theta}$ is diagonal in the number state basis. Then, using the explicit form \eqref{eqn:main-result}, we see that the equation
\begin{eqnarray}
\mathcal{N}_{\gamma}(\hat{U}_{\theta}\hat{\rho}\hat{U}_{\theta}^{\dagger})= \hat{U}_{\theta}\mathcal{N}_{\gamma}(\hat{\rho})\hat{U}_{\theta}^{\dagger}
\end{eqnarray}
has to hold. For the complementary channel, by using the relation (\ref{comp}), we have 
\begin{eqnarray}
\mathcal{N}^{c}(\hat{U}_{\theta}\hat{\rho}\hat{U}_{\theta}^{\dagger})&=&
\Tr_{S}\Big[\rho_{n,m}e^{i\theta (k_{0}(n)-k_{0}(m)) }\ket{n}\bra{m}\nonumber\\
&\otimes & \ket{-i\tau_{n},\lambda^{\pm}}\bra{-i\tau_{m},\lambda^{\pm}}\Big]\nonumber\\
&=&\sum_{n}\rho_{n,n}\ket{-i\tau_{n},\lambda^{\pm}}\bra{-i\tau_{n},\lambda^{\pm}}\nonumber\\
&=&\mathcal{N}^{c}(\rho)\nonumber\\
&=&\hat{U}_{\theta}\mathcal{N}^{c}(\hat{\rho})\hat{U}_{\theta}^{\dagger},
\end{eqnarray}
in which, by using the relation (\ref{eq699}), we define  $K_{0}\ket{n}=k_{0}(n)\ket{n}$, with $k_{0}(n)=\frac{\lambda}{2\omega}n+\frac{1}{2}(1+\frac{\lambda}{2\omega})$. 
\end{proof}

\begin{lemma}\label{lem:optimal-input}
The optimal input state to $\mathcal{N}_{\gamma}$ for the quantum capacity (\ref{eqn:main-result}) is diagonal in the associated basis.
\end{lemma}
\begin{proof}
As the von Neumann entropy is invariant under unitary conjugation, and using the definition (\ref{eq68}) and Lemma \ref{lemma:covariant},  we can easily define: 
\begin{eqnarray}
\hat{\rho}_{\theta}=e^{\mathbbm{i}\hat{K}_{0}\theta}\hat{\rho}e^{-\mathbbm{i}\hat{K}_{0}\theta}
\end{eqnarray}
in such way that we have 
\begin{eqnarray}
J(\hat{\rho}, \mathcal{N}_{\gamma})=J(\hat{\rho}_{\theta}, \mathcal{N}_{\gamma}).
\end{eqnarray}
Since the deformed bosonic dephasing channel is degradable, we can write
\begin{eqnarray}
\int_{0}^{2\pi} d\vartheta P(\vartheta)J(\hat{\rho}_{\vartheta}, \mathcal{N}_{\gamma})\leq J(\int_{0}^{2\pi} d\vartheta P(\vartheta)\hat{\rho}_{\vartheta}, \mathcal{N}_{\gamma}) 
\end{eqnarray}
in which, $\vartheta=2\theta \omega/\lambda $, $P(\vartheta)$ is a probability density. By considering a normalized constant distribution, i.e., $P(\vartheta)=1/2\pi$, for both positive and negative values of $\lambda$, we  obtain 
\begin{eqnarray}
\int_{0}^{2\pi} d\vartheta P(\vartheta)\hat{\rho}_{\vartheta}&=&\sum_{n=0}^{\lfloor 2\nu\rfloor}\hat{\rho}_{nn}\ket{n}\bra{n}, \ \lambda<0,\nonumber\\
\int_{0}^{2\pi} d\vartheta P(\vartheta)\hat{\rho}_{\vartheta}&=&\sum_{n=0}^{\infty}\hat{\rho}_{nn}\ket{n}\bra{n}, \ \lambda>0.\nonumber
\end{eqnarray}
\end{proof}
Therefore, the maximization in the relations  (\ref{gb32}) leads to the maximization over classical probability distributions:
\begin{eqnarray}\label{gb37}
\mathcal{Q}(\mathcal{N}_{\gamma})&=& \max_{p_{n}}\Bigg[S\left(\sum_{n=0}^{d}P_{n}
\ket{n}\bra{n}\right)\nonumber\\
&-&
S\left(\sum_{n=0}^{d}P_{n}
\ket{i \tau_{n},\lambda}\bra{i\tau_{n},\lambda}\right)
\Bigg],
\end{eqnarray}
where $d$ is respectively infinity and $\lfloor 2\nu\rfloor$, when $\lambda>0$ and $\lambda<0$.
Now, we constrain the input average energy, i.e., $\sum_{n=0}^{s}P_{n}\varepsilon^{\pm}(n)\leq E$ with
\begin{eqnarray}
\varepsilon^{\pm}(n)=n\pm \frac{|\lambda|}{2}n^{2}.
\end{eqnarray}
Note that for $\lambda<0$, we should also impose the following condition $s\leq \lfloor 2\nu\rfloor$.
\subsection{Numerical Analysis }\label{subsec:numerical-methods}
\subsubsection{Numerical Analysis for $N=1$}
 For $N=1$,    we consider the following  input state:
\begin{eqnarray}
\Omega=p_{1}\ket{n}\bra{n}+p_{2}\ket{n+m}\bra{n+m},
\end{eqnarray}
where $n$ and $m$ are arbitrary non-negative integers. We can diagonalize the complementary channel term
\begin{eqnarray}
A=\left(
\begin{array}{cc}
    p_{1} & p_{1}\mathcal{K}^{p}_{n,n+m} \\
    p_{2}\mathcal{K}^{p}_{n,n+m} & p_{2}
\end{array}
\right),\nonumber\\
\end{eqnarray}
where $\mathcal{K}^{p}_{m,n}$  is given by (\ref{eqgbv312})-(\ref{eqgbv314})  for negative, zero and positive value of parameter $\lambda$, respectively. We can easily find the eigenvalues:
\begin{eqnarray}
q_{\pm}(\gamma,\lambda)=\frac{1}{2}\Bigg[1\pm\sqrt{(p_{1}-p_{2})^{2}+4p_{1}p_{2}\mathcal{K}^{p}_{n,n+m}}\Bigg].\nonumber
\end{eqnarray}
\subsubsection{Numerical Analysis for $N=2$}
We consider the following  input state, for $N=2$,  
\begin{eqnarray}
\Omega&=&p_{1}\ket{n}\bra{n}+p_{2}\ket{n+m}\bra{n+m}\nonumber\\
&+&p_{3}\ket{n+m+l}\bra{n+m+l}
\end{eqnarray}
where $n$, $m$ and $l$ are arbitrary non-negative integers and $p_{1}+p_{2}+p_{3}=1$. The second term of the relation (\ref{gb37}), can be written as 
\begin{eqnarray}
A=\left(
\begin{array}{ccc}
    A_{11} & A_{12} &  A_{13}  \\
    A_{21} & A_{22} & \Lambda_{23}   \\
     A_{31} & A_{32}&  A_{33}  
\end{array}
\right),
\end{eqnarray}
where
\begin{eqnarray}
A_{11}&=&p_{1},\nonumber\\ 
A_{12}&=&p_{1}\mathcal{K}^{p}_{n,n+m}\nonumber\\
A_{13}&=&p_{1}\mathcal{K}^{p}_{n,n+m+l}\nonumber\\
A_{21}&=&p_{2}\mathcal{K}^{p}_{n+m,n}\nonumber\\
A_{22}&=&p_{2}\nonumber\\
A_{23}&=&p_{2}\mathcal{K}^{p}_{n+m,n+m+l}\nonumber\\
A_{31}&=&p_{3}\mathcal{K}^{p}_{n+m+l,n}\nonumber\\
A_{32}&=&p_{3}\mathcal{K}^{p}_{n+m+l,n+m}\nonumber\\
A_{33}&=&p_{3}\nonumber
\end{eqnarray}
Hence, we are able to calculate the eigenvalues and then the von Neumann entropy. 
Therefore, an optimization method  gives the capacity of the channel.  

\section{Conclusion}
We have shown how to model the impact of a Kerr non-linearity on the evolution of an optical mode, for instance, propagating in an optical fiber for quantum communication applications, when the mode is subject to a non-unitary, continuous-variable, bosonic dephasing channel. 
This channel on its own is an important example of a non-Gaussian channel for which the quantum capacity and certain dependencies on the dephasing rate were known already. In our analytical treatment, based on the notion of nonlinear coherent states and deformed annihilation and creation operators, the quantum capacity is obtained even in the presence of a Kerr medium. The resulting deformed bosonic dephasing channel hence serves as an elegant and convenient way to describe the overall non-unitary, non-Gaussian dynamics that originates from a unitary Kerr evolution combined with non-unitary dephasing. 

Our results show that the quantum capacity including deformation can be greater than that for the undeformed, standard dephasing channel, i.e., the Kerr non-linearity can compensate the dephasing to various extents. Crucially, for this effect to occur, the sign of the Kerr non-linearity must be chosen appropriately which is, in principle, possible by engineering or tailoring the non-linearities as, for instance, occurring in a photonic crystal \cite{41} and so also in a hollow-core photonic crystal fiber \cite{42}. As a next step, the inclusion of photon loss would be of practical relevance \cite{19,32}. Ultimately, an engineering of the interplay between deterministically occurring non-linearities, such as the Kerr effect, and random, non-deterministic noise channels, such as linear loss and nonlinear dephasing, on the hardware level of the optical channels could be supplemented with  active bosonic quantum error correction through codes adapted to the overall error channel evolutions \cite{32}.      


\section*{Acknowledgement} Funding from the Federal Ministry of Education and Research of Germany, project identification number: 16KISQ039 (SD) and the DFG Emmy-Noether program under grant number NO 1129/2-1 (JN) as well as support of the Munich Center for Quantum Science and Technology (MCQST) are acknowledged. 
PvL further acknowledges support from the EU/BMBF via QuantERA (project ShoQC) and from the BMBF in Germany through QR.X.

\appendix
\section{Proof of Lemmata}
\begin{proof}[Proof of Lemma \ref{lem:properties-of-A}]
    the first property follows since $\hat g(x,y)$ is diagonal in the basis $\{\ket{n}\}_{n=0}^d$ for every choice of parameters $x,y$, and by explicit inspection of the respective values of the diagonal entries.\\
    The second property follows again since $\hat g(x,y)$ is diagonal in the basis basis $\{\ket{n}\}_{n=0}^d$ for every choice of parameters $x,y$.\\
    The third property is proven as follows. First, it holds 
    \begin{align}
        \hat A^\dag&=\sqrt{x+y\hat a^\dag\hat a}\hat a^\dag\\
        \hat A\ket{n} 
            &= \sqrt{(x+y)+y\hat a^\dag\hat a}\hat a\ket{n}\qquad\forall n.
    \end{align}
    Therefore we get 
    \begin{align}
        \hat A=\hat a g(x,y)=\hat g(x+y,y)\hat a\label{eqn:formulas-for-A}\\
        \hat A^\dag =\hat g(x,y)\hat a^\dag=\hat a^\dag\hat g(x+y,y).\label{eqn:formulas-for-Adag}
    \end{align}
    Equipped with equations \eqref{eqn:formulas-for-A} and \eqref{eqn:formulas-for-Adag} we can show that
    \begin{align}
        \hat A^\dag \hat A
            &= (x+y)\hat a^\dag\hat a + y\hat a^{\dag2}\hat a^2.
    \end{align}
    To show the fourth property we first calculate $\hat K_0$ explicitly, starting out with the term $\hat A\hat A^\dag$ which evaluates to
    \begin{align}
        \hat A\hat A^\dag 
            &= x+y +(x+3y)\hat a^\dag\hat a+y\hat a^{\dag2}\hat a^2.
    \end{align}
    Therefore it holds
    \begin{align}
        [\hat A,\hat A^\dag] 
            &= x+y+2y\hat a^\dag\hat a\\
            &=\hat g(x+y,2y)^2.
    \end{align}
    Since $\hat K_0=[\hat A,\hat A^\dag]/2$ it then follows that
    \begin{align}
        [\hat K_0,\hat A]
            &= \tfrac{\hat a}{2}\hat g(x,y)(\hat g(x-y,2y)^2-\hat g(x+y,2y)^2)\\
            &= \tfrac{\hat a}{2}\hat g(x,y)(x-y+2y\hat n-(x+y+2y\hat n))\\
            &= -y\hat a\hat g(x,y)\\
            &= -y\hat A
    \end{align}
    where we used $g(x+y,y)\hat a=\hat a g(x,y)$ repeatedly. Likewise, we have 
    \begin{align}
        [\hat K_0,\hat A^\dag] 
            &= y\hat A^\dag.
    \end{align}

\end{proof}

\begin{proof}[Proof of Lemma \ref{lem:decomposition}]
    Let us consider operator $F(t)$, with parameter $t$ defined by
    \begin{eqnarray}
    F(t)=e^{\left[\beta A^{\dagger}+\beta^{\ast}A\right]t}=e^{\zeta_{+}(t) A^{\dagger}}e^{\ln[\zeta_{0}(t)] K_{0}} e^{\zeta_{-}(t) A}
    \end{eqnarray}
    where $\zeta_{\pm}(t)$ and $\zeta_{0}(t)$ are c-number functions of parameter $t$ to be determined under the following conditions:  $\zeta_{\pm}(0) = 0$ and $\zeta_{0}(0) = 1$. Once $\zeta_{\pm}(t)$ and  $\zeta_{0}(t)$ are determined, $\zeta_{\pm}$ and $\zeta_{0}$ are given by $\zeta_{\pm}(1)$ and $\zeta_{0}(1)$, respectively. \\
    Now we have
    \begin{eqnarray}
    \frac{d}{dt}F&=&\left[\beta A^{\dagger}+\beta^{\ast}A\right] F\\
    &=&\dot{\zeta}_{+}\hat{A}^{\dagger}F\nonumber\\
    & &+e^{\zeta_{+}(t) A^{\dagger}}\frac{\dot{\zeta}_{0}}{\zeta_{0}}K_{0}e^{\ln[\zeta_{0}(t)] K_{0}} e^{\zeta_{-}(t) A}\nonumber\\
    & &+e^{\zeta_{+}(t) A^{\dagger}}e^{\ln[\zeta_{0}(t)] K_{0}}\dot{\zeta}_{-}A e^{\zeta_{-}(t) A}\\
    &=& \dot{\zeta}_{+}\hat{A}^{\dagger}F\nonumber\\
    &+& \frac{\dot{\zeta}_{0}}{\zeta_{0}}(K_{0}-\frac{\lambda}{2}\zeta_{+}A^{\dagger})F\label{eqgbqc44}\nonumber\\
    &+& \dot{\zeta}_{-} e^{-\lambda\ln[\zeta_{0}]/2}(A-2\zeta_{+}K_{0}+\frac{\zeta_{+}^{2}\lambda}{2}A^{\dagger})F\label{eqgbqc45}
    \end{eqnarray}
     Hence,  three  coupling differential equations are derived, i.e., 
    \cite{ban1993decomposition}:
    \begin{eqnarray}
    &&e^{-\lambda\ln[\zeta_{0}]/2}\frac{d\zeta_{-}}{dt}=\beta^{\ast},\label{eqgbqc52}\\
    &&\zeta_{0}^{-1}\frac{d\zeta_{0}}{dt}-2e^{-\lambda\ln[\zeta_{0}]/2}\zeta_{+}\frac{d\zeta_{-}}{dt}=0\\
    &&\frac{d\zeta_{+}}{dt}-\frac{\lambda\zeta_{+}}{2\zeta_{0}}\frac{d\zeta_{0}}{dt}+e^{-\lambda\ln[\zeta_{0}]/2}\frac{\lambda\zeta_{+}^{2}}{2}\frac{d\zeta_{-}}{dt}=\beta
    \end{eqnarray}
    The simple calculations gives the  differential equation: 
    \begin{eqnarray}
    \frac{d}{dt}\zeta_{+}+\frac{\lambda}{2}\beta^{\ast}\zeta_{+}^{2}=\beta
    \end{eqnarray}
    which leads to the following:
    \begin{eqnarray}
    \zeta_{+}&=&\sqrt{\frac{\beta}{\beta^{\ast}}}  \sqrt{\frac{2}{\lambda}} \tanh\left[t\sqrt{\frac{\lambda}{2}}|\beta|\right]\\
    &=&e^{i\phi}\sqrt{\frac{2}{\lambda}} \tanh\left[t\sqrt{\frac{\lambda}{2}}|\beta|\right]\label{eqgbqc57}
    \end{eqnarray}
    in which  $\beta=|\beta|e^{i\phi}$. Now by using the equation (\ref{eqgbqc52}), we write
    \begin{eqnarray}\label{eqgbqc58}
    \zeta_{0}=\cosh^{-4/\lambda} \left[t\sqrt{\frac{\lambda}{2}}|\beta|\right].
    \end{eqnarray}
    Now, by considering $t=1$ in the relations (\ref{eqgbqc57}) and (\ref{eqgbqc58}), the desire result are obtained.
 \end{proof}
Note that Baker-Hausdorff lemma and the communication relations (\ref{eqgbqc31})-(\ref{eqgbqc33}) can be used to achieve the following relations: 
\begin{eqnarray}
e^{\zeta_{+}(t) A^{\dagger}}K_{0}e^{\zeta_{-}(t) A}
=K_{0}-\frac{\lambda}{2}\zeta_{+}A^{\dagger}\\
e^{\zeta_{+}(t) A^{\dagger}}e^{\ln[\zeta_{0}(t)] K_{0}}A e^{\ln[\zeta_{0}(t)] K_{0}}e^{\zeta_{-}(t) A}=\zeta_{0}^{-\lambda/2}\nonumber\\
\times \left(A-2\zeta_{+}K_{0}+\frac{\zeta_{+}^{2}\lambda}{2}A^{\dagger}\right)
\end{eqnarray}
used to obtain the relations (\ref{eqgbqc44}) and (\ref{eqgbqc45}), respectively. 

Further note that in the case of $\lambda<0$, by using the fact that $\cosh{ix}=\cos{x}$ and $\sinh{ix}=i\sin{x}$, we can rewrite the relations (\ref{eqgbqc57}) and (\ref{eqgbqc58}) as following:
\begin{eqnarray}
\zeta_{+}&=&\sqrt{\frac{\beta}{\beta^{\ast}}}  \sqrt{\frac{2}{|\lambda|}} \tan\left[t\sqrt{\frac{|\lambda|}{2}}|\beta|\right]\\
&=&e^{i\phi}\sqrt{\frac{2}{|\lambda|}} \tan\left[t\sqrt{\frac{|\lambda|}{2}}|\beta|\right]\label{eqgbqcp57}
\end{eqnarray}
and
\begin{eqnarray}\label{eqgbqcp58}
\zeta_{0}=\cos^{4/|\lambda|} \left[t\sqrt{\frac{|\lambda|}{2}}|\beta|\right].
\end{eqnarray}

\begin{proof}[Proof of Lemma     \ref{lemma:exp(b+bdag)-applied-to-zero}] 
    According to Lemma \ref{lem:decomposition} we have 
    \begin{align}
         e^{-\mathbbm{i}\mu(\hat B+\hat B^\dag)}\ket{0}
            &= e^{\zeta\hat A^\dag}e^{\ln{\zeta_0}\hat K_0}e^{-\zeta*\hat A}\ket{0}\\
            &= e^{\zeta\hat A^\dag}e^{\ln(\zeta_0)(x+y)}\ket{0}\\
            &= e^{\ln(\zeta_0)(x+y)}\sum_{k=0}^\infty \tfrac{(\zeta\hat A^\dag)^k}{k!}\ket{0}\\
            &= e^{\ln(\zeta_0)(x+y)}\sum_{k=0}^\infty \tfrac{(\zeta\hat A^\dag)^k}{k!}\ket{0}.
    \end{align}
    Since $A^\dag\ket{k}=\sqrt{(k+1)(x+y(k+1))}\ket{k+1}$, we have
    \begin{align}
         e^{-\mathbbm{i}\mu(\hat B+\hat B^\dag)}\ket{0}
            &= e^{\ln(\zeta_0)(x+y)}\sum_{k=0}^\infty \tfrac{\zeta^k\sqrt{k!\prod_{i=1}^k(x+i\cdot y)}}{k!}\ket{k}\\
            &= \zeta_0^{(x+y)}\sum_{k=0}^\infty \tfrac{\zeta^k\sqrt{y^k\left(\tfrac{x+y}{y}\right)^{(k)}}}{\sqrt{k!}}\ket{k}.
    \end{align}
    Since by definition (see equation \eqref{def:zeta}) we have $\zeta=\tfrac{-\mathbbm{i}\mu}{|\mu|} y^{-1/2}\tanh(\sqrt{y}|\mu|)$ and $\zeta_{0} = \cosh^{-2/y}\left( \sqrt{y}|\mu|\right)$, 
    the result follows. 
\end{proof}

\begin{proof}[Proof of Lemma \ref{lem:product-of-deformed-coherent-states}]
    Using the Umbral calculus \cite{39,40} the above sum can be lead  to the result.
\end{proof}

\section{Probability Distributions for  Optimization Cases}
We illustrate optimal values of $p_{n}$, as shown in Fig.~\ref{fig:opt} for $N = 1,2,3$, from top to bottom, for different values of the parameter $\lambda$, i.e., $\lambda=-0.1, 0, 0.1$,  from left to right, respectively.
\begin{figure*}[t!]
    \centering
    \includegraphics[width=16cm]{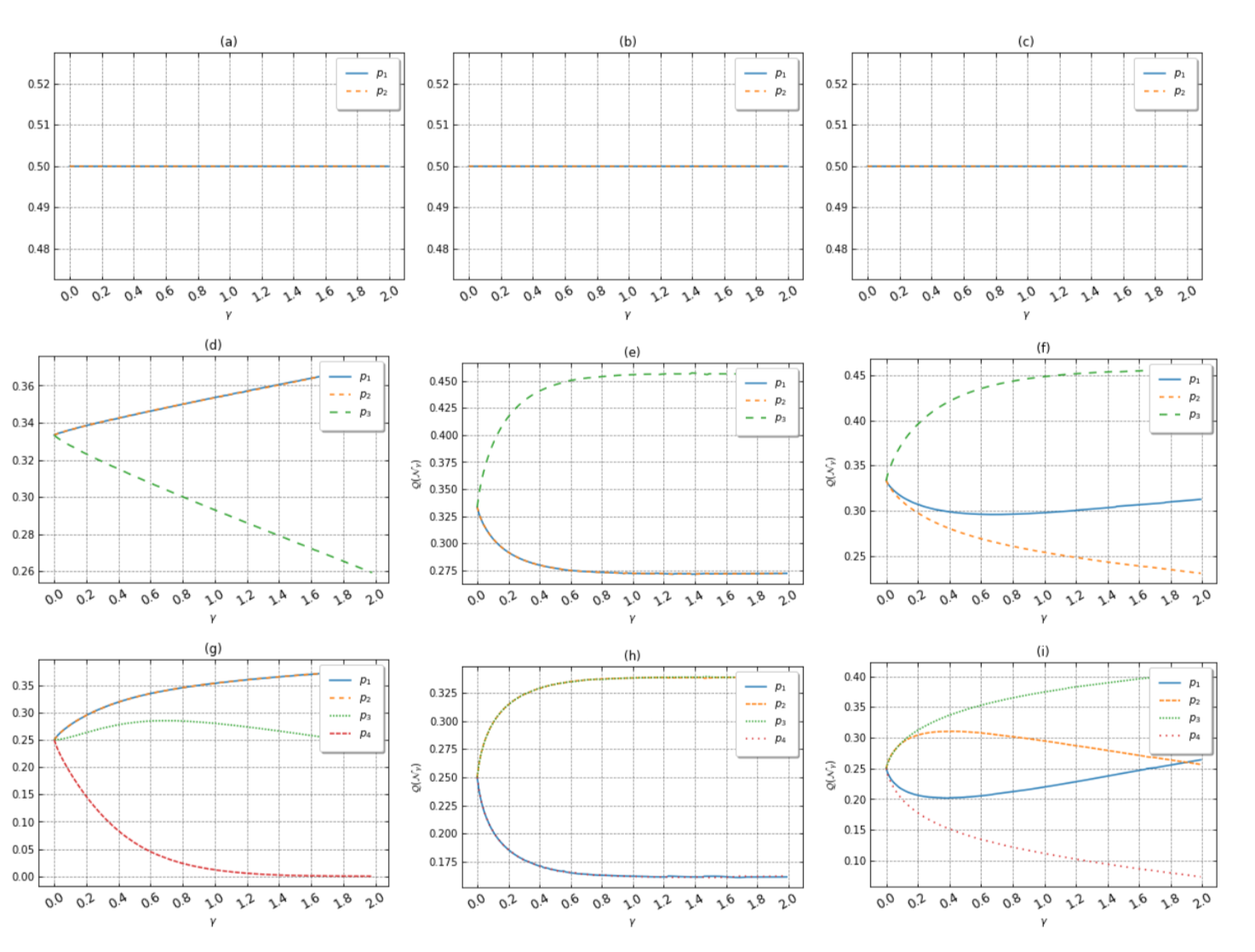}
    \caption{Optimal value of $p_{n}$ for $n = 0,1,\cdots, ,N$ versus $\gamma$; plots (a)-(c) demonstrate $p_{n}$ for different values of $\lambda$, i.e., $\lambda=-0.1, 0, 0.1$ and  $N=1$, respectively; plots (d)-(f) demonstrate $p_{n}$ for the same value of $\lambda$ and  $N=2$ respectively; plots (g)-(i) demonstrate $p_{n}$ for the same value of $\lambda$ and  $N=3$, respectively    }
    \label{fig:opt}
\end{figure*}

\end{document}